\documentclass[a4paper,twocolumn,10pt,accepted=2020-07-06]{quantumarticle}
\pdfoutput=1
\usepackage[utf8]{inputenc}
\usepackage[english]{babel}
\usepackage[T1]{fontenc}
\usepackage{amsmath}
\usepackage{hyperref}

\usepackage{tikz}
\usepackage{lipsum}

\usepackage[cm]{fullpage}

\usepackage{amsmath,bbm}
\usepackage{amsthm}
\usepackage{latexsym}
\usepackage{amssymb}
\usepackage{float}
\usepackage{graphicx}           
\usepackage{color}
\usepackage{relsize}

\newcommand{\be}{\begin{equation}} 
\newcommand{\ee}{\end{equation}}
\newcommand{\beq}{\begin{eqnarray}}
\newcommand{\eeq}{\end{eqnarray}}

\newcommand{\I}{\mathbbm{1}}
\newcommand{\B}{\mathcal{B}}
\newcommand{\A}{\mathcal{A}}
\newcommand{\ket}[1]{|#1\rangle}
\newcommand{\bra}[1]{\langle#1|}

\newcommand{\la}{\langle}
\newcommand{\ra}{\rangle}
\newcommand{\s}{\ket{\tilde{\psi}}}

\newcommand{\blk}{\color{black}}

\newtheorem*{thm}{Theorem}
\newtheorem{lemma}{Lemma}
\newtheorem{result}{Result}
\newtheorem*{assumption}{Assumption}

\begin{document}

\title{Sum-of-squares decompositions for a family of noncontextuality inequalities and self-testing of quantum devices}

\author{Debashis Saha}
\affiliation{Center for Theoretical Physics, Polish Academy of Sciences, Aleja Lotnik\'{o}w 32/46, 02-668 Warsaw, Poland}
\orcid{0000-0003-4525-0903}
\author{Rafael Santos}
\affiliation{Center for Theoretical Physics, Polish Academy of Sciences, Aleja Lotnik\'{o}w 32/46, 02-668 Warsaw, Poland}
\orcid{0000-0002-4695-530X}
\author{Remigiusz Augusiak}
\affiliation{Center for Theoretical Physics, Polish Academy of Sciences, Aleja Lotnik\'{o}w 32/46, 02-668 Warsaw, Poland}
\orcid{0000-0003-1154-6132}

\maketitle

\begin{abstract}
  Violation of a noncontextuality inequality or the phenomenon referred to `quantum contextuality' is a fundamental feature of quantum theory. In this article, we derive a novel family of noncontextuality inequalities along with their sum-of-squares decompositions in the simplest (odd-cycle) sequential-measurement scenario capable to demonstrate Kochen-Specker contextuality. The sum-of-squares decompositions allow us to obtain the maximal quantum violation of these inequalities and a set of algebraic relations necessarily satisfied by any state and measurements achieving it. With their help, we prove that our inequalities can be used for self-testing of three-dimensional quantum state and measurements. Remarkably, the presented self-testing results rely on a single assumption about the measurement device that is much weaker than the assumptions considered in Kochen-Specker contextuality.
\end{abstract}

To realize genuine quantum technologies such as cryptographic systems, quantum simulators or quantum computing devices, the back-end user should be ensured that the quantum devices work as specified by the provider. Methods to certify that a quantum device operates in a nonclassical way are therefore needed. The most compelling one, developed in the cryptographic context, is self-testing \cite{MayersYao}. It exploits nonlocality, i.e., the existence of quantum correlations that cannot be reproduced by the local-realist models, and provides the complete form of device-independent \footnote{With the requirement of the spatial separation between measurements on subsystems, and without any assumption on the internal features of the devices.} characterization of quantum devices only from the statistical data the devices generate.
Thus, it is being extensively studied in recent years \cite{Yang,Bamps,Coladangelo}.

However, since self-testing, as defined in Ref. \cite{MayersYao}, stands on nonlocality \cite{bell} (or, in other words, quantum correlations that violate local-realist inequalities), it is restricted to preparations of composite quantum systems and local measurements on them. Therefore, it poses a fundamental question: presuming the minimum features of the devices how to characterize $(i)$ quantum systems of prime dimension that are not capable of exhibiting nonlocal correlations, and $(ii)$ quantum systems without entanglement or spatial separation between subsystems? A possible way to address such instances is to employ quantum contextuality (Kochen-Specker contextuality), a generalization of nonlocal correlations obtained from the statistics of commuting measurements that are performed on a single quantum system \cite{ks,cabello08,csw,kcbs}. 
%
%
Indeed, the recent study \cite{bharti,knill,bharti2019} provides self-testing statements based on contextual correlations (or correlations that violate noncontextuality inequality).  Since quantum contextual correlations are essential in many aspects of quantum computation \cite{magic,mbc} and communication \cite{horodecki,saha}, self-testing statements are crucial for certifying quantum technology \cite{bharti2019}. Apart from that, it is, nonetheless, fundamentally interesting to seek the maximum information one can infer about the quantum devices only from the observed statistics in a contextuality experiment. 

In the context of nonlocality, sum-of-squares (SOS) decomposition of quantum operators associated with local-realist inequalities has been the key mathematical tool in recent years to obtain optimal quantum values and self-testing properties of quantum devices \cite{Bamps,chainbell,satwap,Kaniewski,sarkar2019selftesting,Augusiak_2019,kaniewski2019weak,cui2019generalization}. Whether this line of study, albeit, restricted to nonlocal correlations, can further be extended to contextuality scenario is of great interest from the perspective of unified approach to non-classical correlations \cite{csw,Amaral2018}.

In this work, we consider Klyachko-Can-Binicio\ifmmode \breve{g}\else \u{g}\fi{}lu-Shumovsky (KCBS) scenario which comprises of one preparation and $n$ (where $n\geqslant 5$ is odd) number of measurements \cite{kcbs,ncycle,lsw}. This is the simplest scenario capable to exhibit contextual correlations using a three-dimensional quantum system and five binary outcome measurements. It also has several implications in quantum foundation and quantum information \cite{guhne,horodecki,qkd,cabello2013,kurzynski,saha-ramanathan,xu-saha}.
We first introduce a modified version of KCBS expression for $n=5$ involving correlation between the outcomes of two sequential measurements, along with an SOS decomposition of the respective quantum operator. We describe our methodology to obtain SOS and simultaneously, generalize for $n$-cycle KCBS scenario where $n=2^m+1, m\in \mathbbm{N}$. Interestingly, the SOS decomposition holds even without the idealizations that the measurements satisfy commutativity conditions in a cyclic order. By virtue of this decomposition, we obtain the maximum quantum value of our modified $n$-cycle expression and a set of algebraic relations involving any quantum state and measurements that yield those maximum values. By solving those relations, we show the existence of a three-dimensional vector-space invariant under the algebra of measurement operators. Subsequently, we prove the uniqueness of the projected three-dimensional measurements and state up to unitary equivalence, that is, self-testing property of the quantum devices.  The presented self-testing statement relies on the premise that the measurement device
returns only the post-measurement system and has no memory, while it does not rely
on the commutativity relations between observables.

\section{Preliminaries} \label{sec:2}
We begin by illustrating our scenario and specifying the assumptions.

\textit{Sequential-measurement set-up.} Each run of the experimental observation comprises of preparation of a physical system followed by two measurements in a sequence using one non-demolishing measurement device as depicted in Fig. \ref{fig}. The measurement device has $n$ (odd) different settings, each of which yields $\pm 1$ outcome. Let's denote the first and second measurement settings by $\A_i$ and $\A_j$ where $i,j\in \{1,\dots,n\}$. The settings are chosen such that $j=i\pm 1$, where from now on the subscript $i$ is taken modulo $n$, that is, $\A_{i\pm n}=\A_i$. We make the following assumption about the measurement device.
\begin{assumption}\label{as1}
The measurement device has no memory and returns only the actual post-measurement state.
\end{assumption}
This assumption is necessary, otherwise, any quantum statistics can be reproduced by classical systems. \\

By repeating this experiment many times we can obtain joint probabilities $p(a_i,a_{i\pm1}|\A_i,\A_{i\pm1})$ of two measurements and single probabilities $p(a_i|\A_i)$ of the first measurement, and consequently, their correlation functions,
\beq \label{prob}
 \la \A_i \A_{i\pm1}\ra &=& \sum_{a_i,a_{i\pm1}} a_ia_{i\pm1} p(a_i,a_{i\pm1}|\A_i,\A_{i\pm1}), \nonumber  \\
 \la \A_i\ra &=& \sum_{a_i} a_i p(a_i|\A_i),
\eeq  
where the measurement outcomes are denoted as $a_i=\pm 1$.

\begin{figure}[http]
\centering
\includegraphics[scale=0.29]{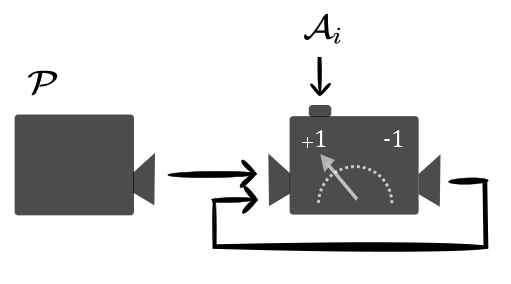}
\caption{{\bf Sequential-measurement set-up.} The simplest contextuality scenario comprises of one preparation $\mathcal{P}$ and one measurement device with settings $\A_i$ each of them returns $\pm 1$ outcome. }
\label{fig}
\end{figure}
In quantum theory the two-outcome measurements $\mathcal{A}_i$ can be in general non-projective. However, since we do not restrict the dimension of these measurements, an extension of 
Naimark's dilation theorem \cite{knill} allows us to consider these measurements to be projective.
Thus, we can represent the measurements by the following operators
\be \label{Ai}
A_i=2P_i-\I ,
\ee 
where $P_i$ are projectors acting on some finite-dimensional Hilbert space $\mathcal{H}$.  
The preparation is represented by a quantum state that, by the same reason, can be considered pure; we denote it by $|\psi\ra$. 

Kochen-Specker contextuality \cite{csw} pertains to the assumption that the projectors satisfy certain orthogonality relations, particularly in this scenario, $P_iP_{i\pm 1}=0$ for all $i$, implying $[A_i,A_{i\pm1}]=0$. Such prerequisite about the measurement device are difficult to justify in practice. Since we aim to characterize the quantum devices from their minimal features, we do not make this assumption. We will see later that orthogonality relations between projectors will be derived facts from the maximal violation of our inequality. 

A general linear expression that can be considered to test nonclassicality (or noncontextuality in the usual scenario) in this set-up is given by,
\be \label{genB}
\B = \sum_i c_i ( \la \A_i\A_{i+1} \ra + \la \A_{i+1}\A_i \ra ) + \sum_i d_i \la \A_i\ra .
\ee 
Using the quantum expression of the joint probabilities under the aforementioned Assumption, for example, $p(+1,+1|\A_i,\A_{i\pm 1}) = \la \psi | P_iP_{i\pm1} P_i|\psi\ra $, we find
\be 
\la \A_i\A_{i+1} \ra + \la \A_{i+1}\A_i \ra = \la \psi | \{ A_i, A_{i+1} \} | \psi \ra .
\ee 
Subsequently, the optimal quantum value of the expression \eqref{genB} is defined as
\be \label{genqB}
\eta^Q = \sup_{|\psi\ra, A_i}  \la \psi|B|\psi\ra, 
\ee
where $B=\sum_i c_i \{ A_i, A_{i+1} \} + \sum_i d_i A_i$ is the quantum operator associated with the expression $\B$ and $A_i$ are of the form \eqref{Ai}. Notice that in the usual scenario, due to commutativity relations, $\{A_i,A_{i+1}\}$ can be replaced by $2A_iA_{i+1}$. The maximal classical value $\eta^C$ (or noncontextual value in the usual scenario \footnote{ Since any noncontextual value assignment pertains to certain orthogonality conditions, here we refer to $\eta^C$ as the classical value for the relaxed scenario. Note that, under the aforesaid Assumption, the optimal value of $\B$ in classical theory or any other theory where measurement does not affect the system is given by Eq. \eqref{gennc}. With the orthogonality conditions, $\eta^C$ reduces to the maximal noncontextual value.}) is defined as
\be \label{gennc}
\eta^C = \max_{a_i\in \{1,-1\}} \left\{2 \sum_{i} c_i a_ia_{i+1} + \sum_{i} d_i a_i \right\}.
\ee

\textit{KCBS inequality.}
The well known $n$-cycle KCBS noncontextuality inequality \cite{ncycle} is of the form 
\be \label{kcbsn}
\B_{\mathrm{KCBS}}:=  - \sum^n_{i=1} \la \A_i\A_{i+1}\ra\leqslant \eta^C = n-2 .
\ee
The maximal quantum violation of this inequality is 
\begin{equation}
    \eta^Q = \frac{3\cos{(\pi/n)}-1}{1+\cos{(\pi/n)}}n
\end{equation}
and it is achieved by the following quantum state
\begin{equation}
|\widehat{\psi} \rangle =\ket{0}\equiv (1,0,0)^{T}, \label{ops} 
\end{equation}
and observables
\begin{equation}
    \widehat{A}_i = 2 |\widehat{v}_{i} \rangle\! \langle \widehat{v}_{i}| - \mathbbm{1}, \label{opm}
\end{equation}
where $\ket{\widehat{v}_{i}}$ are three-dimensional real 
vectors defined as
\be \label{pn}
|\widehat{v}_{i} \rangle = (\cos{\theta},\sin{\theta}\sin{\phi_{i}},\sin{\theta}\cos{\phi_{i}})^{T}
\ee
where $\theta$ is defined as
$\cos\theta=\sqrt{1/(1+2\alpha)}$, where
\be \label{alpha}
\alpha = \frac{1}{2}\sec\left(\frac{\pi}{n}\right)
\ee 
and
\be \label{phi}
\phi_i = \frac{n-1}{n}\pi i .
\ee 
Note that $\alpha$ and $\phi_i$ are functions of $n$, which for the sake of simplification is not explicitly specified in their notation. 
Let us also remark that $|\widehat{\psi}\rangle \in \mathbbm{C}^3$ and $\widehat{A}_i$ acting on $\mathbbm{C}^3$ denote a particular example of quantum realizations achieving the maximal quantum value of the KCBS inequality \eqref{kcbsn}. 
The self-testing properties of the above-mentioned state and measurements based on the violation of KCBS inequality are shown in \cite{bharti}. The proof is based on the optimization method of semidefinite programming under the usual assumptions of contextuality, along with an additional assumption that $P_i$ in Eq. \eqref{Ai} are rank-one projectors.  \\

\textit{Sum-of-squares decomposition.} Let us finally discuss the concept of
sum-of-squares decompositions. Consider a quantum operator $B$ corresponding to some noncontextuality expression $\mathcal{B}$ like the one in \eqref{genqB}.
Now, if for any choice of quantum measurements $A_i$ and some $\eta\in\mathbbm{R}$ one can decompose the shifted operator $\eta\mathbbm{1}-B$ as
\be \label{sosg}
\eta \I - B = \sum_k E^\dagger_k E_k, 
\ee 
the maximal quantum value of $\mathcal{B}$ is upper
bounded by $\eta$, i.e., $\langle \psi|B|\psi\rangle \leqslant \eta$ for any quantum state $|\psi\ra$. We call (\ref{sosg}) a sum-of-squares decomposition associated to $B$. Typically $E_k$ are constructed from the measurement operators $A_i$. The bound $\eta$ is realized by a state and a set of measurements if and only if the following algebraic relation holds true for all $k$,
\be \label{ar}
\quad E_k |\psi\ra = 0.
\ee 
Our self-testing proofs heavily rely on the above relations.

Let us remark that Ref. \cite{lsw} provides an SOS decomposition for the conventional KCBS operator under the assumptions that the measurements satisfy $[A_i,A_{i\pm1}]=0$. In what follows we 
derive an alternative noncontextuality inequality together with the 
corresponding SOS decomposition of the form  \eqref{sosg} 
which does not require making this assumption.
Furthermore, our SOS is designed in such a way that the algebraic relations \eqref{ar} it implies can be used for self-testing.


\section{Modified KCBS inequality with sum-of-squares decomposition}

We are now ready to present our results. For pedagogical 
purposes we begin with the simplest case of $n=5$ and consider
the following modified KCBS expression
\be \label{B5}
\B =  - \frac12 \sum^5_{i=1} (\la \A_i\A_{i+1}\ra + \la \A_{i+1}\A_i \ra ) - \alpha^2 \sum^5_{i=1} \la \A_i\ra, 
\ee 
where $\alpha$ is given in \eqref{alpha} with $n=5$. Following \eqref{gennc} it is not difficult to find the maximal classical value of $B$ is $\eta^C=3+\alpha^2$.

\begin{result} [Modified KCBS inequality with SOS]
The maximal quantum value of $\B$ given in Eq. \eqref{B5} with $\alpha=(1/2)\sec(\pi/n)$ is $\eta^Q=3(1+\alpha^2)$.
\end{result}
\begin{proof}
To prove this statement we present the SOS decomposition for the modified KCBS operator 
\be \label{kcbsop}
B = - \frac{1}{2} \sum_i \{A_i,A_{i+1}\} - \alpha^2\sum_i A_i.\ee
Let us first define the following Hermitian operators for $i=1,\dots,5$, 
\beq
M_{i,1} &=& -\frac{1}{\alpha^3} (A_i+\alpha A_{i-1}+\alpha A_{i+1}), \nonumber  \\ 
M_{i,2} &=& -\frac{1}{\alpha^4} (-\alpha A_i + A_{i-2} + A_{i+2}),
\eeq 
and observe that they satisfy the following relations 
\begin{equation} \label{B5decom}
- \frac{\alpha^5}{5}\sum_i \left(2M_{i,1}+ \alpha^3 M_{i,2}\right)= \alpha^2\sum_i  A_i,   
\end{equation}
and
\begin{equation}
    \frac{\alpha^5}{5} \sum_i \left( M^2_{i,1} + \frac{\alpha^3}{2} M^2_{i,2} \right) = \frac12 \sum_i \{ A_i,A_{i+1} \} +\frac{5}{2\alpha} \I ,
\end{equation}
where we have used the identities $\alpha^2+\alpha=1$ for $\alpha$ given in Eq. \eqref{alpha} with $n=5$ and $A_i^2=\I$.
With the aid of these relations it is straightforward to verify that 
\beq \label{SOS0}
&& \frac{\alpha^5}{5} \sum_i \!\left(\I - M_{i,1} \right)^2 \!+\! \frac{\alpha^8}{10} \sum_i\! \left(\I - M_{i,2} \right)^2  \nonumber \\
&&\hspace{0.5cm}= \left( \alpha^5+\frac{\alpha^8}{2} \right)\I - \frac{\alpha^5}{5}\sum_i \left(2M_{i,1}+ \alpha^3 M_{i,2}\right)  \nonumber \\
&& \hspace{1cm}+ \frac{\alpha^5}{5} \sum_i \left( M^2_{i,1} + \frac{\alpha^3}{2} M^2_{i,2} \right)  \nonumber \\
&&\hspace{0.5cm}= 3(1+\alpha^2) \I - B, 
\eeq  
where $B$ is given in Eq. \eqref{kcbsop}.

Thus, the above equation constitutes a SOS decomposition \eqref{sosg} of the modified KCBS operator in which
\begin{equation}
E_{k} = \sqrt{\frac{\alpha^5}{5}} (\I - M_{k,1})    
\end{equation}
for $k=1,\ldots,5$;
\begin{equation}
E_{k} = \sqrt{\frac{\alpha^8}{10}} (\I - M_{k-5,2})
\end{equation}
for $k=6,\ldots,10$;
and $3+3\alpha^2=4.146$ is the quantum bound of $B$. We can validate that the state and measurements in dimension three \eqref{ops}-\eqref{opm} responsible for optimal value of KCBS inequality achieve this bound.
\end{proof}
Inspired by the above $n=5$ case, let us now derive our modified KCBS expression for more measurements. 
Our aim is to obtain a general expression for which the sum-of-squares decomposition 
can easily be constructed as the one in Eq. \eqref{SOS0}
and later directly used for self-testing.

To reach this goal, let us consider $n$ two-outcome quantum measurements
represented by operators $A_i$ (\ref{Ai}) acting on some 
Hilbert space of unknown but finite dimension. 
Let us then consider the expression \eqref{sosg} in which the operators $E_k$ are of the form $\I-M_k$ with some positive 
multiplicative factors, where $M_k$ are constructed from $A_i$.  Notice that for such a choice,  Eq. \eqref{ar} implies that $M_k$ must be stabilizing operators of the state $|\psi\ra$ maximally violating our modified KCBS expression, that is, $M_k|\psi\ra = |\psi\ra$.
Now, to design the explicit form of $M_k$ we can use the 
optimal quantum realization \eqref{ops}-\eqref{opm} of the $n$-cycle KCBS inequality \eqref{kcbsn},
which gives us (see Appendix \ref{app:stab} for details of the derivation)
\be \label{Mik}
M_{i,k} = \bar{\alpha} \left[ \left(1-2\beta_{k}\right) A_i + \beta_{k}(A_{i+k} + A_{i-k})  \right],
\ee
where $i=1,\dots,n$ and $k=1,\dots,(n-1)/2$, whereas
the coefficients $\beta_{k}$ and $\bar{\alpha}$ are given by 
\be \label{beta} 
\beta_{k} = \frac{1}{2(1-\cos{\phi_{k}})} 
\ee
and
\begin{equation}
\bar{\alpha} = \frac{1+2\alpha}{1-2\alpha},
\end{equation}
where
%
$\alpha$, $\phi_{k}$ are defined in Eqs. \eqref{alpha} and \eqref{phi}, respectively.
%
Let us remark that $M_{i,k},\bar{\alpha}, \beta_i$ are all functions of $n$ which for the sake of simplification is not specified explicitly. Moreover, the operators $M_{i,k}$ defined in \eqref{Mik} act on unknown Hilbert space $\mathcal{H}$ of finite dimension. 

We now go back to the SOS decomposition \eqref{sosg} which is deemed to be of the form 
\begin{equation}
    \sum_{i,k} c_k \left[\I - M_{i,k} \right]^2
\end{equation}
with some non-negative parameters $c_k$ to be determined. 
By plugging the expression of $M_{i,k}$ \eqref{Mik} into it and after some rearrangement of indices, we obtain
\begin{widetext}
\beq \label{sos1}
\sum_{i,k} c_k \left[\I - M_{i,k}\right]^2 & = & \left( n \bar{\alpha}^2   \sum_k c_k  \left( \frac{1}{\bar{\alpha}^2} + 1 + 6\beta_{k}^2 - 4\beta_{k} \right) \right) \I -\left( 2\bar{\alpha} \mathlarger{\sum}\limits_k c_k \right) \sum_i A_i \nonumber \\
&&+\bar{\alpha}^2 \sum_i \left[ 2 c_1 \beta_{1} \left(1-2\beta_{1} \right) +c_{\frac{n-1}{2}}  \beta^2_{\frac{n-1}{2}} \right] \{A_i,A_{i+1}\}  \nonumber \\
&& +  \bar{\alpha}^2  \sum_i  \sum^{(n-3)/2}_{k=2} \left[ 2 c_{k} \beta_{k} \left(1-2\beta_{k}\right) + c_{f\left(\frac{k}{2}\right)}   \beta^2_{f\left(\frac{k}{2}\right)} \right] \{A_i,A_{i+k}\},
\eeq 
\end{widetext}
where
\beq
f\left(\frac{k}{2}\right) = 
\begin{cases}
k/2, & \text{ if $k$ is even} \\
(n-k)/2, & \text{ if $k$ is odd}.
\end{cases}
\eeq
We want to choose the coefficient $c_k$ so that they are non-negative and all the anti-commutators $\{A_i,A_{i+k}\}$ vanish except for $k=\pm1$.
For that purpose we consider $n=2^m+1$ for $m\in \mathbbm{N}\setminus \{1\}$. 
First we take $c_k=0$ whenever $k \neq 2^x$, where $x=0,\dots,m-1$. 
It follows from \eqref{sos1} that our requirement is fulfilled if the following set of equations is satisfied
\be \label{ckeqs}
2 c_{2^{x}} \beta_{2^{x}} \left(1-2\beta_{2^{x}}\right) +c_{2^{x-1}}  \beta_{2^{x-1}}^2 =0
\ee
for $x=1,\dots,m-1$.
The above equation \eqref{ckeqs} implies for all $x=1,\dots,m-1$
\beq \label{ck/c1}
\frac{c_{2^x}}{c_1} &=& \frac{1}{2^x }\prod^x_{j=1} \frac{ \beta_{2^{j-1}}^2 }{ \beta_{2^j} \left(2\beta_{2^j}-1 \right)} \nonumber \\
&=& \left(\frac{\beta_{1}}{2^{x} \beta_{2^{x}}}\right)^2 \prod^x_{j=1} \sec(\phi_{2^j}) .
\eeq 
Since $\sec(\phi_{2^j})$ is positive for all $j$ \footnote{Note that $\cos{\phi_{2^j}}=\cos{(\pi 2^j/n)}$ and $0 < \pi 2^j/n < \pi/2, \forall j=1,2,\dots,m-1$.}, $c_{2^x}/c_1$ is also positive. Now, to provide a plausible solution of $c_{2^x}$, it suffices to choose a positive $c_1$. Due to \eqref{ckeqs} the remaining anti-commutators in \eqref{sos1} are $\{ A_i, A_{i+1}\}$ with a factor
\be \label{sos3}
\bar{\alpha}^2 \left[ 2 c_1 \beta_{1} \left(1-2\beta_{1} \right) + c_{2^{m-1}}  \beta_{2^{m-1}}^2 \right].
\ee
For simplicity we choose this factor to be 1/2 which implies that $c_1$ is such that 
\be \label{c1get}
4 c_1 \beta_{1} \left(1-2\beta_{1} \right) + 2 c_{2^{m-1}}  \beta_{2^{m-1}}^2 =  \frac{1}{\bar{\alpha}^2 }.
\ee 
After substituting $c_{2^{m-1}}$ from Eq. \eqref{ck/c1}, the above gives
\begin{equation} \label{c1}
c_1 = \frac{2^{2m-3}}{\bar{\alpha}^2 }  
 \frac{1}{2^{2m-1}\beta_{1} \left(1-2\beta_{1} \right) + \beta_{1}^2 \ \prod\limits^{m-1}_{j=1} \sec(\phi_{2^j})}.
\end{equation}
One can readily verify that $c_1$ is positive.
Finally, due to \eqref{ckeqs} and \eqref{c1get}, Eq. \eqref{sos1} reads as,
\be \label{sosn}
\sum_{i,k} c_k \left[\I - M_{i,k} \right]^2  = \eta_n \I - B_n,
\ee
where
\be
B_n = - \frac12 \sum_i \{ A_i, A_{i+1} \} - \gamma\sum_i A_i \ , \ee
\be
\gamma = -2 \bar{\alpha} \sum_k c_k \ , 
\ee
and
\beq \label{etan}
\eta_n =  n\bar{\alpha}^2  
\sum_k c_k \left( \frac{1}{\bar{\alpha}^2} + 1+6\beta_{k}^2 - 4\beta_{k} \right),
\eeq
and $c_k, M_{i,k}$ are defined in \eqref{ck/c1}, \eqref{c1} and \eqref{Mik}.

From Eq. \eqref{beta} we know that $\bar{\alpha}$ is a negative quantity and hence $\gamma$ is positive. Thus, our modified $n$-cycle KCBS inequality is 
\begin{equation} \label{Bn}
\B_n := - \frac12 \sum_i ( \la \A_i\A_{i+1}\ra + \la \A_{i+1} \A_i \ra ) - \gamma \sum_i \la \A_i \ra 
\leqslant  \eta_n^C
\end{equation} 
whose quantum bound is $\eta_n$ \eqref{etan} and the classical value $\eta^C_n$ is provided in \textit{Result 3}. It follows from the construction of the SOS \eqref{sosn} that the qutrit quantum state and measurements defined in Eqs. \eqref{ops}-\eqref{phi} satisfy the stabilizing relations $M_{i,k}|\psi\ra=|\psi\ra$, implying the bound $\eta_n$ is tight, or, in other words, the maximal quantum value of
(\ref{Bn}) equals $\eta_n$.

To put the above mathematical analysis in a nutshell, the expression of the noncontextuality inequality \eqref{Bn} is derived such that it meets a SOS decomposition \eqref{sosg} of certain form. This leads us to the following result.

\begin{result}[Modified $n$-cycle expression with SOS]
The maximum quantum value of modified $n$-cycle noncontextuality expression \eqref{Bn} with a SOS decomposition \eqref{sosn} is $\eta_n$ \eqref{etan} (where $n=2^m+1,m\in \mathbbm{N}\setminus \{1\})$.
\end{result}
Let us finally prove the classical bound of our new noncontextuality expression.

\begin{result}[Maximal classical value]
The classical value of $\B_n$ in Eq. \eqref{Bn} is given by $n+\gamma -2$.
\end{result}
\begin{proof}
The classical value can be obtained by assigning $\pm 1$ values to the observables appearing in \eqref{Bn}, that is,
\be \label{Bc}
\eta_n^C = \max_{a_i\in \{1,-1\}} \left\{ - \sum^n_{i=1} a_ia_{i+1} - \gamma \sum^n_{i=1} a_i \right\}, 
\ee
where $\gamma$ is positive.
Let us say in the optimal assignment there are $k$ number of $a_i$ which are $-1$. We first assume $k> n/2$. When there are $k$ number of $-1$, and $n-k$ number of $+1$, the minimum value of $\sum_i a_ia_{i+1} = 4k - 3n$, and the quantity $\sum_i a_i = n-2k$. Substituting these values in \eqref{Bc} we see
\be \label{ncb3}
\eta_n^C = \left(3-\gamma\right) n - \left(4-2\gamma\right) k.\ee
Therefore, the optimal value of $\eta_n^{C}$ is obtained for the minimum value of $k$, that is, for $k=(n+1)/2$. This implies the right-hand-side of \eqref{ncb3} is $n+\gamma -2$. Similarly, if $k<n/2$, then we have $(n-k)>n/2$, and following a similar argument we can obtain the same bound.
\end{proof}

\section{Self-testing of quantum devices}

An exact self-testing statement provides us the certification of quantum devices, given that we observe an optimal violation of a noncontextuality inequality. However, the observed statistics are unchanged in the presence of auxiliary degrees of freedom (or auxiliary systems) and a global unitary. Therefore, self-testing in the context of state-dependent quantum contextual correlation \cite{bharti,knill} infers unique state and measurements up to these equivalences. 

Here, we take the definition of self-testing stated in \cite{knill}. Formally, self-testing of preparation $|\overline{\psi}\ra\in\mathbbm{C}^d$ and a set of measurements $\{\overline{A}_i\}^n_{i=1}$ acting on $\mathbbm{C}^d$ is defined as follows: if a set of observables $\{A_i\}^n_{i=1}$ acting on unknown finite-dimensional Hilbert space $\mathcal{H}$ and a state $|\psi\ra\in\mathcal{H}$ maximally violate a noncontextuality inequality, then there exists 
a projection $\mathbbm{P}:\mathcal{H}\to \mathbbm{C}^d$ and a unitary operation 
$U$ on $\mathbbm{C}^d$ such that
\begin{enumerate}
    \item $U(\mathbbm{P}|\psi\ra) = |\overline{\psi}\ra $ ,
    \item $U(\mathbbm{P}A_i\mathbbm{P})U^{\dagger}=\overline{A}_i$ for all $i=1,\ldots,n$. \blk 
\end{enumerate}
%
%
%
%
%
To obtain self-testing only from the reduced Assumption mentioned in section \ref{sec:2}, we consider a modified version of the expression $\B_n$ $\eqref{Bn}$ of the following form
\be \label{Bn2}
\tilde{\B}_n := \B_n - \sum_i \left[p(++|\A_{i+1},\A_{i}) + p(++|\A_{i-1},\A_{i}) \right].
\ee 
Since the additional term is non-positive, the classical and quantum bounds of $\tilde{\B}_n$ are the same as for $\B_n$. Moreover, it follows from \eqref{sosn} that the SOS decomposition of $\tilde{B}_n$ is
\beq \label{sosn2}
 \eta_n \I - \tilde{B}_n &=& \sum_{i,k} c_k \left[\I - M_{i,k} \right]^2 + \sum_i (P_iP_{i+1})^\dagger(P_iP_{i+1}) \nonumber \\
&& + \sum_i (P_iP_{i-1})^\dagger(P_iP_{i-1}),
\eeq 
where 
\begin{eqnarray}
    \tilde{B}_n = B_n - \sum_i P_{i+1}P_iP_{i+1}  - \sum_i P_{i-1}P_iP_{i-1} ,
\end{eqnarray}
and $\eta_n$ is again the optimal quantum value of $\tilde{B}_n$.
Let us now show that our inequality \eqref{Bn2} can be used to make a self-testing statement, according to the above definition, for the state and observables \eqref{ops}-\eqref{opm} maximally violating it.
\begin{result}[Self-testing]
Under the Assumption stated in Sec. \ref{sec:2}, if a quantum state $\ket{\psi}\in \mathcal{H}$ and a set of $n$ (where $n=2^m+1,m\in \mathbbm{N}\setminus \{1\})$ measurements $A_i$ acting on $\mathcal{H}$ violate the inequality \eqref{Bn2} maximally, then there exists a projection $\mathbbm{P}:\mathcal{H}\to\mathbbm{C}^3$ and a unitary $U$ acting on $\mathbbm{C}^3$ such that  
\beq \label{sn}
& U (\mathbbm{P} A_i \mathbbm{P}^\dagger) U^\dagger = 2|\widehat{v}_i \ra\!\la  \widehat{v}_i | - \I_3, \nonumber \\
& \quad  U(\mathbbm{P}|\psi\ra) = (1,0,0)^T,
\eeq
where $|\widehat{v}_i\ra$ are defined in \eqref{pn}.
%
%
\end{result}
\begin{proof}
Taking the expectation value of the state $\ket{\psi}$ on both side of the SOS decomposition \eqref{sosn2} of $\mathcal{B}$, we obtain by virtue of
\eqref{ar} that for any $i$ and $k$,
\begin{equation}
M_{i,k}\ket{\psi} = \ket{\psi}.    
\end{equation}
In the particular $k=1$ case this condition when combined with 
the explicit form of $M_{i,1}$ given in Eq. \eqref{Mik} together with the fact that $\beta_1 = \alpha/(1+2\alpha)$, leads to the following relations for all $i=1,\dots,n$,
\be \label{c}
(A_i + \alpha A_{i+1} + \alpha A_{i-1} ) \ket{\psi} = (1-2\alpha) \ket{\psi}.
\ee 
Similarly, from the last two terms of the SOS decomposition \eqref{sosn2} we get that for all $i=1,\ldots,n$,
\beq \label{c0}
P_iP_{i\pm1}|\psi\ra  = 0.
\eeq 
Given the relations \eqref{c} and \eqref{c0}, the next Theorem provides the proof for the self-testing statement. 
\end{proof}
The self-testing property implies our modified inequality \eqref{Bn2} are non-trivial since any classical value assignment is not equivalent to the realization given in \eqref{sn}.
\begin{thm}\label{thm}
If a set of quantum observables $\{A_i\}^n_{i=1}$ (where $n$ is odd) of the form \eqref{Ai} acting on arbitrary finite-dimensional Hilbert space $\mathcal{H}$ and a unit vector $\ket{\psi}\in \mathcal{H}$ satisfy the relations \eqref{c} and \eqref{c0}, then there exists a projection operator $\mathbbm{P}:\mathcal{H}\to \mathbbm{C}^3$ and a unitary $U$ acting on $\mathbbm{C}^3$ such that 
\eqref{sn} holds true.
\end{thm}

\begin{proof}
We prove this theorem in two steps. 

\textit{Step 1.} In the first step, we deduce the effective dimensionality of the observables $A_i$ and the state $\ket{\psi}$. Let us define a vector space 
$V = \text{Span} \{\ket{\psi}, A_1\ket{\psi},A_3\ket{\psi}\}.$ 
Due to {\it Lemma} \ref{le:v} (stated in Appendix \ref{app:lem}), it suffices to consider the observables $A_i$ and the state $|\psi\ra$ restricted to $V$.
In other words, {\it Lemma} \ref{le:v} points out that the Hilbert space $\mathcal{H}$ can be decomposed as $V\oplus V^{\bot}$ and all the operators $A_i$ have the following block structure
\be
 A_i =  \left(\begin{array}{@{}c|c@{}}
   \tilde{A}_i  & \mathbb{O} \\
\hline
  \mathbb{O} & A_i'
\end{array}\right),
\ee 
wherein $\tilde{A}_i, A_i'$ are acting on $V,V^{\bot}$, respectively; in particular, $A_i'\ket{\psi}=0$ for any $i$. This allows us to define
\beq
& \tilde{A}_i = \mathbbm{P}A_i \mathbbm{P}^\dagger = 2\tilde{P}_i-\I, \nonumber  \\ 
& \ket{\tilde{\psi}} = \mathbbm{P}\ket{\psi} ,
\eeq
where $\mathbbm{P}$ is the projection operator from $\mathcal{H}$ to $V$, $\tilde{P}_i=\mathbbm{P}P_i \mathbbm{P}^\dagger\geqslant 0$ and $\I$ is the identity operator acting on $V$. 

%
%

It follows from Eq. \eqref{Ai} and Eqs. \eqref{c} and \eqref{c0} that the projected measurements $\tilde{P}_i$ and the state $\ket{\tilde{\psi}}$ satisfy the following sets of relations for all $i=1,\dots,n$,
\beq  
& \quad \tilde{P}_i\tilde{P}_{i\pm 1}\s = 0, \label{cond1} \\
& \left(\tilde{P}_i+\alpha\tilde{P}_{i-1}+\alpha\tilde{P}_{i+1}\right)\s= \s , \label{cond2} 
\eeq 


\textit{Step 2.}
In the second step, we characterize the observables $\tilde{A}_i$. 
With the help of \textit{Lemma} \ref{le:pvm} given in Appendix \ref{app:lem}, we first show that all observables 
$\tilde{A}_i$ are of the form
\begin{equation}\label{Aitilde}
\tilde{A}_i=2|v_i\rangle\!\langle v_i| - \I
\end{equation}
for some normalized vectors $\ket{v_i}\in\mathbbm{C}^3$ such that $\langle v_{i}| v_{i\pm1}\rangle=0$.
The remaining part is the characterization of $\ket{ v_i}$.
By plugging Eq. (\ref{Aitilde}) into Eq. \eqref{cond2}
we obtain that for all $i$,
\begin{equation}\label{relation2}
 (   | v_i\rangle\!\langle v_i|+\alpha | v_{i-1}\rangle\!\langle v_{i-1}|+\alpha | v_{i+1}\rangle\!\langle v_{i+1}|)\s=\s.
\end{equation}
We use the fact that $| v_{i}\ra, |v_{i\pm1}\rangle$ are orthogonal and multiply $\bra{ v_{i-1}}$
and $\bra{ v_{i+1}}$ with Eq. (\ref{relation2}), which lead us to the following equations
\begin{equation} \label{cond7}
   \alpha\langle v_{i-1}| v_{i+1}\rangle\langle v_{i+1}\s=(1-\alpha)\langle v_{i-1}\s
\end{equation}
and
\begin{equation}\label{condition3}
  \alpha\langle v_{i+1}| v_{i-1}\rangle\langle v_{i-1}\s=(1-\alpha)\langle v_{i+1}\s
\end{equation}
for all $i$. By substituting the term $\langle v_{i-1}\s$ from the first equation to the second one, we arrive at the following conditions
\begin{equation}\label{ScProd}
\forall i, \quad    |\langle v_{i-1}| v_{i+1}\rangle|=\frac{1-\alpha}{\alpha}.
\end{equation}
Note that, here we use the fact that $\langle v_{i+1}\s \neq 0$ \footnote{If $\protect{\langle v_{j+1}\s} = 0$ for some $j$, then \eqref{cond7} implies $\protect{\langle v_{j-1}\s}$ is also 0, and further \eqref{relation2} implies $\protect{|v_{j}\rangle\!\langle v_{j}\s = \s}$. Substituting these in \eqref{relation2} taking $i=j+1$, 
we arrive at a relation $\protect{|v_{j+2}\rangle\!\langle v_{j+2} \s= (1-\alpha)/\alpha \s}$ which cannot be true for any finite $n$ since $\protect{|v_{j+2}\rangle\!\langle v_{j+2}|}$ has eigenvalues 1,0.}.
Considering the absolute value of both side of \eqref{condition3}
and using \eqref{ScProd} we obtain another set of conditions
\begin{equation} \label{ScProd2}
\forall i, \quad    |\langle \tilde{\psi}| v_{i-1}\rangle|=|\langle \tilde{\psi}| v_{i+1}\rangle| .
\end{equation}
%
And since $n$ is odd, as a consequence of the above equation, 
\begin{equation} \label{ScProd2a}
\forall i,j, \quad    |\langle \tilde{\psi}| v_{i}\rangle|=|\langle \tilde{\psi}| v_{j}\rangle| .
\end{equation}
Let us try to see what is the most general form of $\ket{ v_i}$
compatible with the above conditions. 
First let us exploit the fact that observed probabilities do not change if we rotate the state and measurements by a unitary operation. 
We thus choose it so that
$U\s=(1,0,0)^T\equiv\ket{0}$.
We also notice that any unitary of the following form 
\begin{equation}\label{unitary}
    \left(
    \begin{array}{cc}
        1 & 0\\
        0 & U'
    \end{array}
    \right)
\end{equation}
with $U'$ being any $2\times 2$ unitary does not change
$\ket{0}$. Later we will use this freedom.

Due to the fact that we are characterizing projectors $\ket{ v_i}\!\langle  v_i|$ rather than the vectors themselves, we can always assume the first element of the vector is positive, that is, $\ket{v_i}$ has the form,
\be \label{vigen}
\ket{ v_i} = \left( \cos\theta_i, e^{\mathbbm{i}a_i}\sin\theta_i\sin\phi_i, e^{\mathbbm{i}b_i}\sin\theta_i\cos\phi_i \right)^T.
\ee 
The condition \eqref{ScProd2a} implies that all $\cos\theta_i$ are equal and therefore let us denote $\theta_i = \theta$. Plugging these forms of $\ket{v_{i}}$ and $\s=\ket{0}$ into Eq. \eqref{relation2}, the first element of the vector equation leads to
\be \label{costheta}
\cos\theta=\frac{1}{\sqrt{1+2\alpha}}.
\ee

Using this freedom we can bring one of the vectors, say $\ket{ v_n}$, 
to $(\cos\theta,0,\sin\theta)^T$ by taking
\begin{equation}\label{phin}
\sin\phi_n=0, \quad e^{\mathbbm{i}b_n}=1 .
\end{equation}
Then, due to the condition $\langle v_1| v_n\rangle=\langle v_{n-1}| v_n\rangle=0$ we infer $e^{\mathbbm{i}b_{1}} ,e^{\mathbbm{i}b_{n-1}}$ are real and without loss of generality we can take
\begin{equation} \label{phij}
 e^{\mathbbm{i}b_{1}} = e^{\mathbbm{i}b_{n-1}} = 1
\end{equation}
by absorbing the sign in $\cos\phi_1,\cos\phi_{n-1}$.
Further, we can get rid one of the phases in $\ket{ v_1}$, that is,
\be \label{a1}
e^{\mathbbm{i}a_1}=1,
\ee 
and take $\sin(\phi_1)$ to be non-negative by applying another unitary of the form \eqref{unitary},
\begin{equation}
    U'=\mathrm{diag}[\pm\exp(-\mathbbm{i}a_1),1] 
\end{equation}
that does not change the simplified form of $\ket{v_n}$.
Equating the second and third element of the vector equation \eqref{relation2}, we obtain the relations
\begin{equation} \label{w1}
e^{\mathbbm{i}a_i} \sin\phi_i + \alpha e^{\mathbbm{i}a_{i-1}} \sin\phi_{i-1}+ \alpha e^{\mathbbm{i}a_{i+1}} \sin\phi_{i+1} = 0,
\end{equation}
and
\be \label{w2}
e^{\mathbbm{i}b_i} \cos\phi_i + \alpha e^{\mathbbm{i}b_{i-1}} \cos\phi_{i-1} + \alpha e^{\mathbbm{i}b_{i+1}} \cos\phi_{i+1} = 0 .
\ee
With the aid of \eqref{phin} and \eqref{a1}, Eq. \eqref{w1} for $i=n$ points out $\sin(\phi_1) = -e^{\mathbbm{i}a_{n-1}} \sin(\phi_{n-1})$ which allows us to consider $e^{\mathbbm{i}a_{n-1}}=1$.
Taking $i=1$ in Eqs. \eqref{w1} and \eqref{w2} and replacing the values of  $\sin\phi_n,\cos\phi_n,e^{\mathbbm{i}a_1},e^{\mathbbm{i}b_1},e^{\mathbbm{i}b_n}$ we obtain,
\beq 
& \sin\phi_1+\alpha e^{\mathbbm{i}a_2} \sin\phi_2 = 0, \\
& \cos\phi_1+\alpha+\alpha e^{\mathbbm{i}b_2} \cos\phi_2 = 0 .
\eeq 
Thus, $e^{\mathbbm{i}a_2},e^{\mathbbm{i}b_2}$ are real and can be taken to be 1. Note, here we use the fact that $\sin{\phi_1} \neq 0$ \footnote{If $\sin{\phi_1} = 0$, then $\cos{\phi_1} = \pm 1$ and consequently $\protect{\langle v_n | v_{1} \rangle = \cos{(\theta \mp \theta)}}$ which contradicts the relation $\protect{\langle v_n | v_{1} \rangle = 0}$. Analogously, if we suppose $\cos{\phi_2}=0$, then $\cos{\phi_1}+\alpha = 0$ and $\sin\phi_2=\pm1$. Now, the first equation holds only if $2\alpha^2 = 1$.}. Similarly, by taking $i=2,\dots,n-2$ we conclude for all $i$
\be \label{aibi} 
e^{\mathbbm{i}a_i}=e^{\mathbbm{i}b_i}=1.
\ee
On the other hand, the condition $\langle  v_i| v_{i+1}\rangle =0$ implies,
\beq \label{varphii}
\phi_{i+1} - \phi_i &=& \cos^{-1}\left(-\frac{\cos^2\theta}{\sin^2\theta} \right) \nonumber \\
&=& \frac{(n-1)\pi}{n}.
\eeq 
Finally, considering $i=n$ in the above Eq. \eqref{varphii} and using $\sin\phi_n=0$ we deduce $\phi_1=(n-1)\pi/n$. We discard the possibility $\phi_1=-(n-1)\pi/n$ since $\sin\phi_1$ is taken to be non-negative. Thus, the equations \eqref{costheta}, \eqref{aibi}, and \eqref{varphii} together with $\phi_1$ establish that the unknown vectors $\ket{ v_i}$ in \eqref{vigen} are unitarily equivalent to $\ket{\widehat{v}_i}$. This completes the proof. 
\end{proof}

\section{Conclusion}

Kochen-Specker contextuality captures the intrinsic nature of quantum theory that essentially departs from classicality. It also offers a generalization of quantum correlations beyond nonlocality to a larger class of quantum systems and minimizes the demands to test non-classicality. Therefore, it is a fundamental problem to understand what is the maximal information about the underlying quantum system that can be inferred from the correlations observed in a contextuality experiment, and whether this information can be used for certification of quantum devices from minimal assumptions of their internal functioning. 

In this work, we derive self-testing statements for 
$n$-cycle scenario using weaker assumptions than those made in previous approaches based on Kochen-Specker contextuality \cite{csw,bharti,knill,bharti2019}. In particular, we do not assume orthogonality relations between measurement effects. Instead, we consider general two-outcome measurements which 
nevertheless obey a single assumption that the measurement device does not return any additional information except the post-measurement system  and does not possess any memory.
Moreover, we take a different approach, that is, we
use the sum-of-squares 'technique' that has successfully been 
used in the Bell scenario to derive maximal quantum violation of certain Bell inequalities as well as in making self-testing statements \cite{Bamps,chainbell,satwap,Kaniewski,sarkar2019selftesting,cui2019generalization,kaniewski2019weak,Augusiak_2019}, but has never been explored for self-testing in the contextuality scenario. 

We further remark that self-testing from quantum contextuality is not fully device-independent as far as its original definition is concerned, while, its experimental test does not require space-like separation. The assumption is critical to verify for practical purposes, however, in future studies, one may try to overcome it by restricting the computational power or the memory of the measurement device. Nonetheless, it is way more powerful than the usual process of tomography. It is also distinct from the self-testing approach in prepare-and-measure scenario \cite{STtavakoli,STfarkas} since no restriction on the dimensionality of the preparation is imposed here. 

Although the SOS decompositions hold for a certain number of measurements, a suitable adaptation of our approach in future studies may lead to SOS decompositions for an arbitrary odd number of measurements. 
Another direction for further study is to explore whether our approach can be applied to states and measurements of higher dimension than three and whether our self-testing statements can be made robust to experimental imperfections. 
From a more general perspective, it would be interesting to design a
unifying approach to self-testing based on Bell nonlocality and quantum contextuality.

\section*{Acknowledgement} 
This work is supported by the Foundation for Polish Science through the First Team project (First TEAM/2017-
4/31) co-financed by the European Union under the European Regional Development Fund.


\bibliographystyle{alphaarxiv}

\bibliography{bibliography}

\onecolumn
\appendix

\section{Obtaining the stabilizing operators}
\label{app:stab}

To guess the stabilizing operators $M_{i,k}$ we use the stabilizing operators in the optimal quantum realization of $n$-cycle KCBS inequality \eqref{kcbsn}. Let us assume that these operators are in the following form
\begin{equation}
    \widehat{M}_{i,k} = a \widehat{A}_i + b \widehat{A}_{i+k} + b' \widehat{A}_{i-k},
\end{equation}
where the coefficients $a$, $b$ and $b'$ are to be determined
as a solution to the equation
\be \label{guess}
(a \widehat{A}_i + b \widehat{A}_{i+k} + b' \widehat{A}_{i-k}) |\widehat{\psi} \rangle = |\widehat{\psi} \rangle ,
\ee
and $|\widehat{\psi}\rangle,\widehat{A}_i$ are given in  Eqs. \eqref{ops}-\eqref{opm}.
To solve the above we first notice the following relation,
\be
\widehat{A}_i |\widehat{\psi} \rangle = (\cos{2\theta},\sin{2\theta}\sin{\phi_i},\sin{2\theta}\cos{\phi_i})^{T},
\ee
which when substituted into Eq. \eqref{guess} 
leads one to a system of equations
\be
\label{linsysstab1}
\begin{bmatrix} 
  a (1+\frac{b}{a}+\frac{b'}{a}) \cos{2\theta}  \\[1ex]
  a \sin{2\theta} \big( \sin{\phi_i} +\frac{b}{a} \sin{\phi_{i+k}} +\frac{b'}{a} \sin{\phi_{i-k}} \big) \\[1ex]
  a \sin{2\theta} \big( \cos{\phi_i} +\frac{b}{a} \cos{\phi_{i+k}} +\frac{b'}{a} \cos{\phi_{i-k}} \big) \\
  \end{bmatrix} =
  \begin{bmatrix} 
  1  \\ 
  0  \\ 
  0  \\
  \end{bmatrix} . 
\ee
Assuming that $a \neq 0$ and taking into account that $\sin{2\theta} \neq 0$, the last two equations in the above system can be rewritten as 
\beq
\label{linsysstab2}
\begin{bmatrix} 
   \sin{\phi_i} &  \sin{\phi_{i+k}} & \sin{\phi_{i-k}}  \\[1ex]
   \cos{\phi_i} &  \cos{\phi_{i+k}} & \cos{\phi_{i-k}} \\
  \end{bmatrix}
  \begin{bmatrix}
    1 \\ b/a \\ b'/a
  \end{bmatrix} = 
  \begin{bmatrix}
     0 \\ 0
  \end{bmatrix}. 
\eeq
After multiplying the above equation from left by 
\be 
\begin{bmatrix}
  \sin{\phi_i} & \cos{\phi_i} \\
  \cos{\phi_i} & -\sin{\phi_i} \\
\end{bmatrix}
\ee 
and using the fact $\phi_{i+k}-\phi_i = \phi_k$, Eq. \eqref{linsysstab2} simplifies to,
\beq \label{lins3}
\begin{bmatrix} 
   1 &  \cos{\phi_k} & \cos{\phi_k}  \\ 
   0 &  \sin{\phi_k} & -\sin{\phi_k} \\
  \end{bmatrix}
  \begin{bmatrix}
    1 \\ b/a \\ b'/a
  \end{bmatrix} = 
  \begin{bmatrix}
     0 \\ 0
  \end{bmatrix} . 
\eeq
In this way we remark that the dependence of $i$ in \eqref{linsysstab2} disappears and the system of equations \eqref{lins3} imply
\be
\frac{b}{a} = \frac{b'}{a} = -\frac{1}{2} \sec{\phi_k}.
\ee
Substitution of above in the first vector equality of \eqref{linsysstab1} leads to 
\be
a = \frac{1}{(1-\sec{\phi_k})(2\cos^2{\theta}-1)},
\ee
and thus, we obtain a unique solution of $a,b,b'$.
Finally, substituting $a,b,b'$ into Eq. \eqref{guess} we can conveniently state $\widehat{M}_{i,k}$ operators in the following way
\begin{eqnarray}
\widehat{M}_{i,k} := \left(\frac{1+2\alpha}{1-2\alpha}\right) \left[ (1-2\beta_{k}) \widehat{A}_i + 
\beta_{k}(\widehat{A}_{i+k} + \widehat{A}_{i-k})\right],\nonumber\\
\end{eqnarray}
where 
\begin{equation}
\beta_{k} = \frac{1}{2(1-\cos{\phi_k})},
\qquad
\alpha = \frac12 \sec\left(\frac{\pi}{n}\right) .
\end{equation}

\section{Lemma \ref{le:v}-\ref{le:pvm}} \label{app:lem}

In this appendix, we provide two Lemmas that are used in the proof of the Theorem.

\begin{lemma}
\label{le:v}
If a set of quantum observables $\{A_i\}^n_{i=1}$ (where $n$ is odd) of the form \eqref{Ai} and a vector $\ket{\psi}$ satisfy the relations \eqref{c} and \eqref{c0}, then the vector space
\be \label{v}
V =   \mathrm{span} \{\ket{\psi}, A_1\ket{\psi},A_3\ket{\psi}\} \ee
is invariant under the algebra generated by $A_i$.
\end{lemma}
\begin{proof}
To prove this statement it suffices to show that $A_i\ket{\psi}$ for all $i=1,\ldots,n$ as well as all $A_iA_j\ket{\psi}$ with $i\neq j$ can be expressed as linear combinations of the basis vectors $\ket{\psi}$, $A_1\ket{\psi}$ and $A_3\ket{\psi}$. 

Let us begin by noting that 
Eq. \eqref{c} for $i=2$ gives us directly such a linear combination for $A_2\ket{\psi}$ and so $A_2\ket{\psi}\in V$. Then, the fact that $A_i\ket{\psi}\in V$ for $i=4,\ldots,n$ follows from Eq. (\ref{c}); it is enough to rewrite the latter as 
\begin{equation}\label{Raimat}
 A_i\ket{\psi}=\frac{1-2\alpha}{\alpha}\ket{\psi}-\frac{1}{\alpha}A_{i-1}\ket{\psi}-A_{i-2}\ket{\psi}.
\end{equation}

Let us now move on to showing that $A_iA_j\ket{\psi}\in V$
for all $i\neq j$. To this end, we first observe that using 
\eqref{c0} we obtain
\beq 
A_iA_{i\pm 1} \ket{\psi} &=& (2P_i-\I)(2P_{i\pm 1}-\I)\ket{\psi} \nonumber \\
&=& -(A_i + A_{i\pm 1} + \I ) \ket{\psi},
\eeq 
which due to the fact that $A_i\ket{\psi}\in V$, allows us to conclude that for all $i$, $A_iA_{i\pm 1}\ket{\psi}\in V$. 

Let us then consider the vectors $A_iA_j\ket{\psi}$ for pairs $i,j$
such that $|i-j|=2$. Using the property of involution and the fact $[A_i,A_{i\pm1}]\ket{\psi}=0$ which is a consequence of Eq. \eqref{c0}, we get
\beq
    A_i A_{i\pm 2}|\psi\ra &=& A_i A_{i\pm 2} (A_{i \pm 1})^2 |\psi\ra \nonumber \\
    &=& (A_i A_{i\pm 1})(A_{i\pm 1} A_{i\pm 2}) |\psi\ra.
\eeq
Since we have already shown $A_iA_{i\pm 1}\ket{\psi}\in V$, the above equation implies $A_iA_{i\pm 2}\ket{\psi}\in V$.

Given that $A_{i}A_j\ket{\psi}\in V$ for $|i-j|=1$ and $|i-j|=2$ we can 
then prove, applying the same argument as above, that $A_{i}A_j\ket{\psi}$ 
belong to $V$ for any pair $i,j$ such that $|i-j|=3$. In fact, following this approach recursively we can prove that 
$A_iA_j\ket{\psi}\in V$ for $i,j$ such that $|i-j|=k$ 
with $k=3,\ldots,n-1$, which completes the proof.
%
\end{proof}

Let us remark that the subspace $V$ is in fact spanned by any triple 
of the vectors $\ket{\psi}$, $A_i\ket{\psi}$ and $A_j \ket{\psi}$ with $i \neq j$.
This is a consequence of the fact that, as proven above, any vector $A_i\ket{\psi}$ is a linear combination of $\ket{\psi}$, $A_1\ket{\psi}$ and $A_3\ket{\psi}$.

\begin{lemma}
\label{le:pvm}
If a set of projectors $\{\tilde{P}_i\}^n_{i=1}$ acting on $\mathbbm{C}^3$ and a vector $\s$ satisfy the relations \eqref{cond1} and \eqref{cond2}, then each $\tilde{P}_i$ has rank one, that is, for each $i$ there exists a normalized vector $\ket{v_i}\in\mathbbm{C}^3$ such that
$\tilde{P}_i=|v_i\rangle\!\langle v_i|$ and, moreover, $\la v_i|v_{i\pm1}\ra =0$.
\end{lemma}
\begin{proof}
Since $\tilde{P}_i$ are projectors, we have 
\be \label{cond3}
\forall i, \ \tilde{P}_i^2 \s = \tilde{P}_i \s .
\ee 
Let us begin by showing that $\tilde{P}_{i}\s\neq 0$ for all $i$. Assume to this end 
that there exist $j$ such that $\tilde{P}_{j}\s= 0$. Using then 
Eq. \eqref{cond2} for $i=j-1$ we arrive at
\be \label{fs0}
(\tilde{P}_{j-1}+\alpha \tilde{P}_{j-2})\s = \s .
\ee 
After applying $\tilde{P}_{j-2}$ to both sides of this equation and using Eq. (\ref{cond1}), we obtain $\alpha \tilde{P}^2_{j-2}\s=\tilde{P}_{j-2}\s$
which is consistent with Eq. \eqref{cond3} if and only if $\tilde{P}_{j-2}\s=0$. Therefore, due to Eq. \eqref{fs0} we have $\tilde{P}_{j-1}\s=\s$.  Again, substituting these relations in \eqref{cond2} taking $i=j$, we arrive at $\tilde{P}_{j+ 1}\s=[(1-\alpha)/\alpha]\s$ which contradicts Eq. \eqref{cond3}.

Let us now show that all the operators $\tilde{P}_i$ are of rank one. 
We first prove that none of them can be of rank three. Assume for this purpose
that $\mathrm{rank}(\tilde{P}_j)=3$ for some $j$. Then, the condition \eqref{cond3} gives $\tilde{P}_j\s=\s$. This, after taking into account that $\tilde{P}_{j+1}\tilde{P}_j\s=0$ implies $\tilde{P}_{j+1}\s=0$, which contradicts the fact $\tilde{P}_i\s\neq 0$ for all $i$, as shown before. 

Let us then prove that none of $\tilde{P}_i$ can be of rank two. 
To this end, assume that there is $j$ such that $\mathrm{rank}(\tilde{P}_j)=2$
and consider the eigen-decomposition of $\tilde{P}_j$,
\be
\label{Fj}
 \tilde{P}_j =  | 1\ra\!\la  1| + | 2\ra\!\la  2|,
\ee 
where $| 1\ra,| 2\ra,| 3\ra$ are the eigenvectors, forming an orthonormal basis
in $\mathbbm{C}^3$. Subsequently, $\s$ can be expressed as
\be 
\s = x_1| 1\ra + x_2| 2\ra + x_3| 3\ra 
\ee 
for some $x_1,x_2,x_3\in \mathbbm{C}$.
Note that $x_1=x_2=0$ is not possible since it requires $\tilde{P}_j\s=0$. Similarly, $x_3\neq 0$, otherwise $\tilde{P}_j\s=\s$ which implies $\tilde{P}_{j\pm1}\s=0$. 

Now, employing the fact that $\tilde{P}_j$ is supported on $\mathrm{span}\{\ket{ 1},\ket{ 2}\}$, it follows from the condition $\tilde{P}_j\tilde{P}_{j\pm 1}\s=0$ that $\tilde{P}_{j\pm 1}\s= q_{3,\pm}\ket{ 3}$ for some $q_{3,\pm}\in \mathbbm{C}$.
By combining this with (\ref{cond3}) we find that 
\begin{equation}
    \tilde{P}_{j\pm 1}\ket{ 3}=\ket{ 3},
\end{equation}
that is, $\ket{ 3}$ is the eigenvector of $\tilde{P}_{j\pm1}$ with eigenvalue one, 
which, due to the fact that $\tilde{P}_{j\pm 1}\leqslant \mathbbm{1}$, implies that
$\tilde{P}_{j\pm 1}$ decompose as
\begin{equation}\label{Fj1}
    \tilde{P}_{j\pm1}=\tilde{P}_{j\pm1}'+\ket{ 3}\!\bra{ 3}
\end{equation}
with $\tilde{P}_{j\pm1}'$ being projectors supported on $\mathrm{span}\{\ket{ 1},\ket{ 2}\}$. 
By finally plugging Eqs. \eqref{Fj} - \eqref{Fj1} into Eq. 
\eqref{cond2} for $i=j$ and projecting the obtained equation onto $\ket{ 3}$
we see that $2\alpha=1$, which is not satisfied for any $n$. 

As a result all the operators $\tilde{P}_{i}$ are of rank one 
and therefore they can be expressed as 
\begin{equation}
  \tilde{P}_i= |v_i\ra\!\la v_i| 
\end{equation}
for some $\ket{v_i}\in\mathbbm{C}^3$. Furthermore, since $\tilde{P}_i\s \neq 0$, Eq. \eqref{cond1} implies $\la v_i|v_{i\pm1}\ra =0$. This completes the proof.
\end{proof}

\end{document}